\documentclass[12pt]{article}

\usepackage{amsxtra,amssymb,amsthm,amsmath,latexsym}
\usepackage{graphicx}
\usepackage{booktabs}
\usepackage{multirow}
\usepackage{afterpage}
\usepackage{setspace}
\usepackage{microtype}
\usepackage{calc,fourier}
\usepackage[T1]{fontenc}
\usepackage[hidelinks]{hyperref}
\usepackage{float}
\usepackage[small]{caption}
\usepackage{tikz}
\usetikzlibrary{positioning}
\allowdisplaybreaks

\newtheorem{thm}{Theorem}[section]

\newtheorem{theorem}{Theorem}[section]
\newtheorem{lemma}[theorem]{Lemma}

\newtheorem{corollary}[theorem]{Corollary}
\numberwithin{equation}{section}

\newcommand{\bee}{\begin{equation*}}
\newcommand{\eee}{\end{equation*}}
\newcommand{\be}{\begin{equation}}
\newcommand{\ee}{\end{equation}}
\newcommand{\ba}{\begin{align}}
\newcommand{\ea}{\end{align}}

\newcommand{\RRR}{\mathbb{R}^3}

\title{Scattering of EM waves by many small perfectly conducting or impedance bodies}
\author{A.G. Ramm \\
\small Kansas State University, Manhattan, KS 66506-2602, USA\\
\small \texttt{ramm@math.ksu.edu}
}

\date{}
\begin{document}
\maketitle

\begin{abstract}
A theory of electromagnetic (EM) wave scattering by many small particles of an arbitrary shape is developed.
The particles are perfectly conducting or impedance. For a small impedance particle of an arbitrary shape
 an explicit analytical formula is derived for the scattering amplitude. The formula holds
 as $a\to 0$, where $a$ is a characteristic size of the small particle and the wavelength is arbitrary but fixed.
The scattering amplitude for a small impedance particle is shown to be proportional to $a^{2-\kappa}$,
where $\kappa\in [0,1)$ is a parameter which can be chosen by an experimenter as he/she wants. The boundary
impedance of a small particle is assumed to be of the form $\zeta=ha^{-\kappa}$, where $h=$const, Re$h\ge 0$.
The scattering amplitude for a small perfectly conducting particle is proportional to $a^3$, it is much smaller than
that for the small impedance particle.

The many-body scattering problem is solved under the physical assumptions $a\ll d\ll \lambda$, where $d$ is the minimal distance
between neighboring particles and $\lambda$ is the wavelength. The distribution law for the small impedance particles is
$\mathcal{N}(\delta)\sim\int_{\delta}N(x)dx$ as $a\to 0$. Here $N(x)\ge 0$ is an arbitrary continuous function that can be chosen by
the experimenter and $\mathcal{N}(\delta)$ is the number of particles in an arbitrary sub-domain $\Delta$.
It is proved that the EM field in the medium
where many small particles, impedance or perfectly conducting, are distributed, has a limit, as $a\to 0$ and
a differential equation is derived for the limiting field.

On this basis the recipe is given for creating materials with a desired refraction coefficient by embedding
many small impedance particles into a given material.

\end{abstract}


\section{Introduction} \label{sec1}
Electromagnetic (EM) wave scattering is a classical area of research. Rayleigh stated in 1871. see \cite{Ray},  that  the main
 part of the field, scattered by a small body, $ka \ll 1$, where $k$ is the wave number and  $a$ is the characteristic size of
  the body, is the dipole radiation, but  did not give formulas for calculating this radiation for bodies of arbitrary shapes.
   For spherical bodies Mie (1908) gave a solution to EM wave scattering problem using separation of variables in the spherical
    coordinates. This method does not work for bodies of arbitrary shapes. Rayleigh and Mie concluded that EM field, scattered by
     a small body, is proportional to $O(a^3)$. We prove that the field scattered by a small impedance body (particle) of an
     arbitrary shape
is proportional to $a^{2-\kappa}$, where $\kappa\in [0,1)$ is a parameter which can be chosen by the experimenter as he/she wishes,
 see
formula  \eqref{eq1.2} below. Since $2-\kappa<3$,  it follows, for $a\to 0$, that the scattering amplitude for small impedance
 particle is much larger than the scattering amplitude for perfectly conducting or dielectric small particle. This conclusion may
 be of practical importance.

There is a large literature on low-frequency wave scattering and multiple scattering, see \cite{AK},\cite{D},
\cite{L}, \cite{M},  \cite{T}.

In this paper a theory of EM wave scattering by perfectly conducting and by impedance small bodies of arbitrary shapes is developed.
For one-body scattering problem explicit formulas for the scattering amplitudes are derived for perfectly conducting and for
impedance small bodies of arbitrary shapes. For many-body scattering problem the solution is given as a sum of explicit terms
 with the coefficients that solve a linear algebraic system. If the size of the small bodies $a \to 0$ and their number
  $M=M(a) \to \infty$, a limiting integral equation is derived for the field in the limiting medium. This equation allows
   us to obtain a local differential equation for the field in the limiting medium and to give explicit analytic formulas for
    the refraction coefficient of the limiting medium.

As a result we formulate a recipe for creating materials with a desired refraction coefficient by embedding many small impedance particles in a given material.

The methods developed in this paper were applied to acoustic problems in \cite{R632}, to heat transfer
in the medium where many small bodies are distributed in \cite{R624}, to wave scattering by many nano-wires
in \cite{R634}.

In Section \ref{sec2} the theory of EM wave scattering is developed for small perfectly conducting bodies (particles) of arbitrary
 shapes.

In Section \ref{sec3} the theory is developed for EM wave scattering by one impedance particle of an arbitrary shape.

In Section \ref{sec4} the theory is developed for EM wave scattering by many small impedance particles of an arbitrary shape.

In Section \ref{sec5} a recipe for creating materials with a desired refraction coefficient is given is given. The problem of
 creating materials with a desired magnetic permeability is solved.

Physical assumptions in this paper can be described by the inequalities:
\be \label{eq1.1}
	a \ll d \ll \lambda,
\ee
where $\lambda$ is the wavelength in $\RRR \setminus \Omega$,\,\, $\Omega$ is a bounded domain in which many small particles $D_m$ are distributed, $1 \le m \le M=M(a)$, $d$ is the minimal distance between neighboring particles.

The boundary impedance is assumed to be
\be \label{eq1.1'}
	\zeta_m = \frac{h(x_m)}{a^\kappa},
\ee
where $x_m \in D_m$ is an arbitrary point inside $D_m$, $h(x)$ is an arbitrary continuous function in $\Omega$ such that
Re$h \ge 0$, $\kappa \in [0,1)$ is a parameter. {\em  One can choose $h$ and $\kappa$ as one wishes.}

The distribution of the small impedance  particles in $D$ is given by the formula
\be \label{eq1.2}
	\mathcal{N}(\Delta):= \frac{1}{a^{2-\kappa}}\int_{\Delta} N(x)dx(1+o(1)), \quad a \to 0,
\ee
where $\Delta \subset \Omega$ is an arbitrary open set, $\mathcal{N}(\Delta)$ is the number of small particles in the set $\Delta$,
and $N(x) \ge 0$ is an arbitrary continuous function in $\Omega$.

{\em The experimenter can choose the function $N(x)\ge 0$ as he/she wishes.}

One has
\be \label{eq1.3}
	\mathcal{N}(\Delta) = \sum_{x_m \in \Delta} 1.
\ee
By $\omega$ the frequency is denoted, $k=\frac{\omega}{c}$ is the wave number, $c$ is the velocity of light in the air.

\section{Scattering by perfectly conducting particles.} \label{sec2}
\subsection{Scattering by one particle} \label{sec2.1}
The problem is to find the solution to Maxwell's equations
\be \label{eq2.1}
	\nabla\times E = i\omega \mu H, \quad \nabla \times H=-i\omega \epsilon E, \quad
	 \text{ in }D':= \RRR \setminus D,
\ee
where $D$ is the small body, $ka \ll 1$, $a=0.5$diam$D$, $\epsilon$ and $\mu$ are dielectric and magnetic constants of the medium in
 $D'$, $k=\omega \sqrt{\epsilon \mu}$, and the boundary condition is:
\be \label{eq2.2}
	[N,[E,N]]=0 \quad \text{ on } S:=\partial D.
\ee
Here and below $N:=N_s$ is the unit normal to $S$ pointing into $D'$, $[E,N]=E \times N$  is the vector product of two vectors,
 $E\cdot N=(E,N)$ is the scalar product, $|S|$ is the surface area.

The incident field $E_0$ is:
\be \label{eq2.3}
	E_0=\mathcal{E} e^{ik\alpha\cdot x}, \qquad H_0=\frac{\nabla\times E_0}{i\omega \mu},
\ee
where $\alpha \in S^2$ is a unit vector, the direction of the incident plane wave, and it is assumed
that $ \mathcal{E}\cdot \alpha=0$. This assumption implies that
\be \label{eq2.4}
	\nabla\cdot E_0=0, \qquad \nabla \cdot H_0=0.
\ee
The field $E$ to be found is:
\be \label{eq2.3'}
	E=E_0+v_E,
\ee
where the scattered field $v_E$ satisfies the radiation condition
\be \label{eq2.5}
	r\left(\frac{\partial v_E}{\partial r}-ikv_E\right)=o(1),\quad r:=|x|\to \infty.
\ee
In equation \eqref{eq2.5} the $o(1)$ is uniform with respect to the direction $\beta:=\frac {x} {r}$ of the
scattered field as $r\to \infty$.

The scattering amplitude $ A(\beta, \alpha, k)$ is defined as usual:
\be \label{eq2.6}
	v_E=\frac{e^{ikr}}{r} A(\beta, \alpha, k)+o\left(\frac{1}{r}\right), \quad r = |x|\to \infty,\quad \beta=\frac{x}{r}.
\ee
The magnetic field $H=H_0+v_H$,
\be \label{eq2.7}
	H=\frac{\nabla \times E}{i\omega \mu}, \quad v_H=\frac{\nabla\times v_E}{i\omega \mu}.
\ee
Let us look for the solution to the scattering problem \eqref{eq2.1}-\eqref{eq2.5} of the form:
\be \label{eq2.8}
	E=E_0+\nabla\times \int_S g(x,t)J(t)dt, \qquad g(x,t)=\frac{e^{ik|x-t|}}{4\pi|x-t|},
\ee
where $J$ is a {\em tangential field to $S$.} We assume that $S\in C^2$, that is, $S$ is twice continuously differentiable.

Equations \eqref{eq2.1} are satified if
\be \label{eq2.9}
	\nabla\times\nabla\times E=k^2 E, \qquad H=\frac{\nabla\times E}{i\omega \mu}.
\ee
Since $E_0$ satisfies equations \eqref{eq2.9}, these equations are equivalent to
\be \label{eq2.10}
	\nabla\times\nabla\times v_E=k^2 v_E, \qquad v_E=\frac{\nabla\times v_E}{i\omega \mu}.
\ee
Equation for $v_E$ is equivalent to the equations:
\be \label{eq2.11}
	(\nabla^2+k^2)v_E=0, \qquad \nabla\cdot v_E=0 \text{ in }D',
\ee
because $\nabla\times\nabla\times v_E=\nabla\nabla\cdot v_E-\nabla^2 v_E$ and $\nabla\cdot v_E=0$. Conversely, equations \eqref{eq2.11} are equivalent to \eqref{eq2.9} and to \eqref{eq2.1}.

The radiation condition is satisfied by
	$$v_E=\nabla\times \int_S g(x,t)J(t)dt$$
for any vector-function $J(t)$.

 The boundary condition \eqref{eq2.2} yields
\be \label{eq2.12}
	\frac{J}{2} +TJ:=\frac{J}{2} +\int_S [N_s,[\nabla_s g(s,t),J(t)]]dt=-[N_s,E_0],
\ee
where the formula
\be \label{eq2.13}
	\lim_{x\to s^-}[N,\nabla\times\int_S g(x,t)J(t)dt]=\frac{J(s)}{2}+TJ,
\ee
was used, see \cite{R635}.
Let us prove that equation \eqref{eq2.12} has a solution and this solution is unique in the space $C(S)$ of continuous on $S$
functions. This proves that the scattering problem can be solved by formula \eqref{eq2.8} with $J$ solving \eqref{eq2.12}.
\begin{thm} \label{thm2.1}
If $D$ is sufficiently small, then equation \eqref{eq2.12} is uniquely solvable in $C(S)$ and its solution $J$ is tangential to $S$.
\end{thm}
\begin{proof}
Note that any solution to equation \eqref{eq2.12} is a tangential to $S$ field. To see this, just take
the scalar product of $N_s$ with both sides of equation \eqref{eq2.12}. This yields $N_s\cdot J(s)=0$.
In other words, $J$ is a tangential to $S$ field.

Let us check that the operator $T$ is compact in $C(S)$. This follows from the formula
\be \label{eq2.14}
	TJ=\int_S \left(\nabla_S g(s,t)N_s \cdot J(t)-J(t)\frac{\partial g(s,t)}{\partial N_s} \right)dt.
\ee
Indeed, if $J$ is a tangential to $S$ field then
\be \label{eq2.15}
	N_s\cdot J(s) = 0.
\ee
Since $S \in C^2$, relation \eqref{eq2.15} implies
\be \label{eq2.16}
	|N_s\cdot J(t)|=O(|s-t|)|J(t)|, \quad |\nabla_s g(s,t)N_s \cdot J(t)| \le O\left(\frac{1}{|s-t|}\right)|J(t)|.
\ee
Thus, the first integral in \eqref{eq2.14} is a weakly singular compact operator in $C(S)$. The second integral in \eqref{eq2.14}
is also a weakly singular compact operator in $C(S)$ because
\be \label{eq2.17}
	\left|\frac{\partial g(s,t)}{\partial N_s}\right|=O\left(\frac{1}{|s-t|}\right),
\ee
if $S\in C^2$.

Consequently, equation \eqref{eq2.12} is of Fredholm type in $C(S)$. The corresponding homogeneous equation has only the trivial
solution if $D$ is sufficiently small. This follows from the following argument.
The homogeneous version of equation \eqref{eq2.12} means that the function
$$v_E=\nabla\times \int_S g(x,t)J(t)dt$$
 solves equations \eqref{eq2.11}, satisfies the radiation condition \eqref{eq2.5}, and
\be \label{eq2.18}
	[N,v_E]=0 \quad \text{ on }S.
\ee
This implies that $v_E=0$ in $D'$.

 Lemma 2.1 (see below) implies that if  $v_E=0$ in $D'$
 then $J=0$.
This conclusion and the Fredholm alternative prove the existence and  uniqueness of the solution to  equation \eqref{eq2.12}.
The smallness of the body $D$ guarantees that $k^2$ is not a Dirichlet eigenvalue of the Laplacian in $D$.
Theorem 2.1 is proved.
\end{proof}

{\bf Lemma 2.1.} {\em
Assume that the following conditions hold:

a) $v_E=0$ in $D'$,

b) $J$ is tangential to $S$,

and

c) $k^2$ is not a Dirichlet eigenvalue of the Laplacian in $D$.

 Then $J=0$.
 }
$$ $$
{\em Proof.} Denote $A:=\int_S g(x,t)J(t)dt$ and use the formula
\be \label{eq2.181}
	\int_{D'} \nabla\times A\cdot B dx=\int_{D'}A\cdot \nabla \times B dx-\int_S N\cdot[A,B] ds=\int_{D'}A\cdot \nabla \times B dx,
\ee
valid for any $B \in C_0^\infty(D')$. If $\nabla\times A=0$ in $D'$, then formula \eqref{eq2.181} yields
\be \label{eq2.182}
	\int_{D'} A\cdot \nabla\times B dx=0, \quad \forall B\in C_0^\infty(D').
\ee
Write this formula as
\be \label{eq2.183}
	\int_S dt J(t)\cdot \int_{D'}g(x,t)F(x)dx=0, \quad F:=\nabla\times B.
\ee
The set of vector-fields $F$ coincide with the set of divergence-free fields $\nabla\cdot F=0$ in $D'$,where  $F \in C_0^\infty(D')$.

The set
 of vector-fields
 $$G(t)=\int_{D'}g(x,t)F(x)dx,\quad \forall F\in C_0^\infty(D'),$$
 where it is not assumed that the condition $\nabla \cdot F=0$ holds, {\em is dense in the set $L^2(S)$ of vector fields}. Indeed, if there exists an $h \neq 0$ such that
\be \label{eq2.184}
	\int_S h(t)\int_{D'}g(x,t)F(x)dx dt=0, \quad \forall F\in C_0^\infty(D'),
\ee
and $w(x):=\int_S g(x,t)h(t)dt$, then
$$ \int_{D'}w(x) F(x)dx=0, \quad   \forall F\in C_0^\infty(D').$$
Thus,
\be \label{eq2.185}
	w(x)=\int_S g(x,t)h(t)dt=0 \quad \text{ in }D'.
\ee
Consequently,
\be \label{eq2.186}
	(\nabla^2+k^2)w=0 \quad \text{ in }D, \quad w=0 \text{\,\, on \,\,}S.
\ee
Since $k^2$ is not a Dirichlet eigenvalue of the Laplacian in $D$, equation \eqref{eq2.186} implies $w=0$ in $D$.
 Therefore, $w=0$ in $D\cup D'$. This implies $h=\frac{\partial w}{\partial N_+}-\frac{\partial w}{\partial N_-}=0$.
 Consequently,  the set $G(t)$ is dense in the set $L^2(S)$ of vector fields on $S$.

We claim that if $\nabla\cdot F=0$ in $D'$, where $F \in C_0^\infty(D')$, then $\nabla\cdot G=0$ on $S$.

Indeed,
\be \label{eq2.187}
	\nabla_t\cdot \int_{D'} g(x,t)F(x)dx=-\int_{D'} \nabla_x g(x,t)\cdot F(x)dx=\int_{D'}g(x,t)\nabla\cdot F(x)dx=0.
\ee
Conversely, if $\nabla\cdot G=0$ on $S$, then equations  \eqref{eq2.187} show that
$$\int_{D'}g(x,t)\nabla\cdot F(x)dx=0, \forall t \in S.$$

Let us use the local coordinate system with the axis $x_3$ directed along the outer normal $N_s$ to $S$, and $x_1(s), x_2(s)$
are coordinates along two orthogonal axes tangential to $S$. Let us denote by $e_1(s)$ and $e_2(s)$ the unit
vectors along these axes at a point $s\in S$.

Equation \eqref{eq2.183} can be written as
\be \label{eq2.188}
	 \int_S J(t)\cdot G(t)dt=0
\ee
for all smooth $G(t)$ such that $\nabla\cdot G=0$ on $S$,
$G=\int_{D'}g(x,t)F(x)dx,$ $\nabla \cdot F=0$.

Let $J(t)=J_1(t)e_1(t)+J_2(t)e_2(t)$ in the local coordinates. For an arbitrary small $\delta>0$ one can choose $G_1(t)$ and $G_2(t)$ such that
\be \label{eq2.189}
	||\overline{J}_1-G_1||_{L^2(S)}+||\overline{J}_2-G_2||_{L^2(S)} < \delta,
\ee
where the over-bar denotes the complex conjugate. With $G_1$ and $G_2$ so chosen, choose  $G_3$ such that
\be	\label{eq2.190}
	\nabla\cdot G=0 \text{\,\,\, on\,\,\,} S,
\ee
which is clearly possible. Then equation \eqref{eq2.188} yields
\be	\label{eq2.191}
	\int_S(|J_1|^2+|J_2|^2)dt = O(\delta).
\ee
Since $\delta >0$ is arbitrary small, relation \eqref{eq2.191} implies $J_1=J_2=0$. Therefore,  $J=0$.

Lemma 2.1 is proved. \hfill$\Box$

As was stated above, it follows from Lemma 2.1 and from the Fredholm alternative that equation \eqref{eq2.12}
 is uniquely solvable for any right-hand side if $k^2 \not\in \sigma(\Delta_D)$, that is, if $k^2$ is not a Dirichlet eigenvalue
 of the Laplacian in $D$. If $D$ is sufficiently small, which we assume since $a\to 0$, then a fixed number $k^2$ cannot
 be a Dirichlet eigenvalue  of the Laplacian in $D$ because the smallest Dirichlet eigenvalue  of the Laplacian in $D$
 is $O(\frac 1 {a^2})>k^2$ if $a\to 0$.
$$ $$
{\bf Remark 2.1}. {\em The assumption $k^2 \not\in \sigma(\Delta_D)$ can be discarded if $g(x,t)$ is replaced by
 $g_{\epsilon}(x,t)$, the Green function of the Dirichlet Helmholtz operator in the exterior of a ball
 $B_\epsilon:=\{x: |x|\le \epsilon\}$, where $\epsilon>0$ is chosen so that $k^2 \not\in \sigma(\Delta_{D\setminus B_\epsilon})$.
 This choice of $\epsilon>0$ is always possible (see \cite{R190}, p. 29).}

Let us denote by $V$ the operator that gives the tangential to $S$ component $v_{E\tau}$ of the unique solution $v_E$ to
the scattering problem \eqref{eq2.1}--\eqref{eq2.3}, \eqref{eq2.5}:
\be \label{eq2.192}
E=E_0+v_E,\qquad	v_{E\tau}=V(-[N,E_0]).
\ee
If the tangential component $v_{E\tau}$ is known, then $v_E$ is uniquely defined in $D'$. This is a known fact,
see, for example, \cite{R635}.
The operator $V$ is linear and bounded in $C(S)$. It maps $C(S)$ onto $C(S)$ and $v_E$ has the same smoothness as the data $[N,E_0]$.
For example, if $S \in C^\ell$, then $v_E \in C^\ell(D')$, where $\ell>0$.

Define
\be \label{eq2.19}
	Q:=\int_S J(t)dt.
\ee
From formulas \eqref{eq2.6}, \eqref{eq2.8} and \eqref{eq2.19} it follows that
\be \label{eq2.20}
	A(\beta,\alpha,k)=\frac{ik}{4\pi}[\beta,Q].
\ee
For body $D$ one has
\be \label{eq2.21}
	\int_S [N,E_0]ds =\int_D \nabla\times E_0 dx=\nabla\times E_0\, |D| = \nabla\times E_0 \, c_D a^3,
\ee
where $|D|$ is the volume of $D$ and $c_D>0$ is a constant depending on the shape of $D$.
For example, if $D$ is a ball of radius $a$, then $c_D=\frac{4\pi}{3}$.

One has the formula (see \cite{R635}, p.8):
\be \label{eq2.22}
	-\int_S \frac{\partial g(s,t)}{\partial N_S}ds=\frac{1}{2}+o(1), \quad a\to 0.
\ee

Since $N_s \cdot J(s) = 0$ and $S$ is $C^2-$smooth, it follows that $|N_s \cdot J(t)| \leq c|s - t||J(t)|$. Therefore
\be \label{eq2.35f}
I := \left| \int_S ds\int_S dt \nabla_s g(s,t) N_s \cdot J(t) \right| \leq c\int_S ds\int_S dt\frac{1}{|s - t|}|J(t)|,
\ee
and $I \leq O(a) \int_S |J(t)|dt$. If $I$ would satisfy the estimate $I = o(Q)$, as $a \to 0$, then the theory would simplify
 considerably and one would have $Q = - \nabla \times E |D|= - \nabla \times E c_D a^3$.
   Unfortunately, estimate $I = o(Q)$ is not valid, and one has to give a new estimate for the
   integral $I_1:= \int_S ds\int_S dt \nabla_s g(s,t) N_s \cdot J(t)$. To do this, integrate equation    \eqref{eq2.12}
   over $S$, use  equations   \eqref{eq2.14} and   \eqref{eq2.22}, and get

\be\label{eq2.36f}
Q+I_1=-c_Da^3 \nabla \times E_0.
\ee
Let us write $I_1$ as
\be\label{eq2.37f}
I_1=e_p\int_S \Gamma_{pq}(t)J_q(t)dt,
\ee
where $\{e_p\}_{p=1}^3$ is an orthonormal basis of $\mathbb{R}^3$,
\be\label{eq2.38f}
\Gamma_{pq}(t):= \int_S \frac {\partial g(s,t)}{\partial s_p}N_q(s)ds,
\ee
and the integral in formula \eqref{eq2.38f}  is understood as a singular integral.
Thus, equation \eqref{eq2.36f} takes the form
\be\label{eq2.39f}
(I+\Gamma)Q=-c_Da^3 \nabla \times E_0.
\ee
Here the constant matrix $\Gamma$ is determined from the relation
\be\label{eq2.40f}
\Gamma Q=e_p\int_S \Gamma_{pq}(t)J_q(t)dt,
\ee
the summation is understood over the repeated indices $p,q$, so $\Gamma$
is the matrix which sends a constant vector $Q$ onto the constant vector
$I_1$ defined by the equation  \eqref{eq2.37f}.

One can prove that the constant matrix $\Gamma$ exists and can be determined by equation
 \eqref{eq2.40f}, and the matrix $I+\Gamma$ is non-singular.

 To prove that a constant matrix  $\Gamma$ exists
 assume that for every $p=1,2,3,$ the set of functions $\{\Gamma_{pq}(t)\}_{q=1}^3$ is linearly independent in
 $L^2(S)$, $\int_S \Gamma_{pq}^2(t)dt\neq 0$ and $Q=\int_S J(t)dt\neq 0$. Here $J(t)=\sum_{q=1}^3 e_qJ_q(t)$.
 For a fixed $p$ let $M_p$ be the set in $L^2(S)$ orthogonal to the linear span of $\Gamma_{pq}(t)$.
 Then every function $J_q(t)$ can be represented as $J_q(t)=J^0_q(t)+ \sum_{j=1}^3c_{qj}\Gamma_{pj}(t)$,
 where $J^0_q\in M_p$ and $c_{qj}$ are constants. One has
 $$\sum_{q=1}^3 \int_S\Gamma_{pq}(t)J_q(t)=\sum_{q,j=1}^3c_{qj}\int_S \Gamma_{pq}(t)\Gamma_{pj}(t)dt:=
 \sum_{q,j=1}^3 c_{qj}\gamma_{p;qj},$$
 where $\gamma_{p;qj}$ is a constant non-singular matrix for each $p$ because the set $\{\Gamma_{pq}(t)\}_{q=1}^3$
 is assumed linearly independent. To satisfy equation   \eqref{eq2.40f} one has to satisfy the following
 equation:
  $$\sum_{q,j=1}^3 c_{qj}\gamma_{p;qj}=\sum_{q=1}^3\Gamma_{pq}Q_q.$$
  Since we assumed that $Q\neq 0$, at least one of the numbers $Q_q\neq 0$.
  If there is just one such number, say,  $Q_{q_1}\neq 0$ and $Q_q=0$ for $q\neq q_1$, then we set
 $\Gamma_{pq_1}=Q_{q_1}^{-1}\sum_{q,j=1}^3 c_{qj}\gamma_{p;qj}$,  $Q=e_{q_1}Q_{q_1}$ where
 there is no summation over $q_1$, and  $\Gamma_{pq}=0$ for $q\neq q_1$.
 If, for example, $Q_{q_b}\neq 0$, $b=1,2$, then we may set
 $\Gamma_{pq_b}=\frac 1 2 Q_{q_b}^{-1}\sum_{q,j=1}^3 c_{qj}\gamma_{p;qj}$ and  $\Gamma_{pq}=0$ for $q\neq q_b$, $b=1,2$.
 If $Q_q\neq 0$ for $q=1,2,3,$ then we may set  $\Gamma_{pq}=\frac 1 3 Q_{q}^{-1}\sum_{q,j=1}^3 c_{qj}\gamma_{p;qj}$.

  A more physical choice of $\Gamma_{pq}$ is the following one:
  $$ \Gamma_{pq}:= \frac {\overline{Q_q}}{\sum_{m=1}^3 |Q_{m}|^2}\sum_{b,j=1}^3 c_{bj}\gamma_{p;bj}, \qquad \sum_{m=1}^3 |Q_{m}|^2 >0.$$
  The corresponding to this choice  weights are $ \frac {\overline{Q_q}}{\sum_{m=1}^3 |Q_{m}|^2}$, so that
  $\sum_{q=1}^3\Gamma_{pq}Q_q=\sum_{b,j=1}^3 c_{bj}\gamma_{p;bj}$.

  A simpler approach to finding $\Gamma=(\Gamma_{pq})$, which automatically leads to a diagonal matrix $\Gamma= \gamma I$ with
  a number $\gamma$ and the identity matrix $I$,  is to find $\gamma$ from the condition
  $|e_p\int_S\Gamma_{pq}(t)J_q(t)dt-ce_p\int_S J_p(t)dt|=min$, where the minimization is taken over the number $c$
  and $|\cdot|$ is the length of a vector.
  The solution of this minimization problem is $c_{min}:=\gamma=\frac { \sum_{p=1}^3 \overline{Q_p}X_p}{ \sum_{p=1}^3|Q_p|^2}$,
  where $X_p:=\int_S\Gamma_{pq}(t)J_q(t)dt$. For this choice
  of $\Gamma$ one has $(I+\Gamma)^{-1}=(1+\gamma)^{-1} I$.

  From the computational point of view it is simpler to use the formula with the diagonal $\Gamma$,
  to calculate the number  $c_\gamma:=(1+\gamma)^{-1}$, and to calculate the $Q$ by the formula
  $$Q=-c_\gamma c_D  a^3 \nabla \times E_0, \qquad c_\gamma:=(1+\gamma)^{-1}.$$

  The existence of the
 constant matrix $\Gamma_{pq}$ in equation   \eqref{eq2.40f} is proved.


  To prove the second claim, namely, that    the matrix $I+\Gamma$ is non-singular, it is sufficient
  to prove that $\dim R(I+\Gamma)=3$, where $R(B)$ is the range of the matrix $B$. The range of the matrix
  $I+\Gamma$ consists of the vectors $-c_Da^3 \nabla \times E_0$. Let us check
  that the range of the set of vectors $ \{\nabla \times E_0\}$ equals to $3$,
  $\dim  \{\nabla \times E_0\} = 3$, where $E_0=\mathcal{E}e^{ik\alpha \cdot x}$, $\alpha\cdot \mathcal{E}=0$,
  $\alpha\in S^2$ and $\mathcal{E}$ runs through the set of arbitrary constant vectors.
   Since $\nabla \times E_0=ik[\alpha, \mathcal{E}]e^{ik\alpha \cdot x}$ and one can
   obviously choose three pairs of vectors $\mathcal{E}, \alpha$ such that the three vectors
   $[\alpha, \mathcal{E}]$ are linearly independent and  $\alpha\cdot \mathcal{E}=0$, the second claim is proved.

Since the matrix $I+\Gamma$ is non-singular, equation  \eqref{eq2.39f} yields a formula for $Q$:
\be\label{eq2.41f}
Q=-c_Da^3 (I+\Gamma)^{-1}\nabla \times E_0.
\ee

Let us formulate the result using the simplified diagonal form of the matrix $\Gamma$.

{\bf Theorem 2.2.} {\em
One has
\be\label{eq2.43f}
Q =-c_Da^3 c_\gamma \nabla \times E_0, \quad a \to 0, \qquad c_\gamma:=(1+\gamma)^{-1}.
\ee
}
To use this result practically one has to solve numerically the integral equation \eqref{eq2.12} for $J$,
calculate $Q:=\int_SJ(t)dt:=\sum_{p=1}^3 e_pQ_p$, then calculate $\gamma= \frac { \sum_{p=1}^3 \overline{Q_p}X_p}{ \sum_{p=1}^3|Q_p|^2}$,
where $X_p:=\int_S\Gamma_{pq}(t)J_q(t)dt$, and then use formula  \eqref{eq2.43f}.
\vspace {4mm}

From formulas \eqref{eq2.3}, \eqref{eq2.20} and \eqref{eq2.43f} one calculates $A(\beta, \alpha, k)$.

\subsection{Many-body scattering problem.} \label{eq2.2}
Let $D_m$, $1 \le m \le M=M(a)$ be small perfectly conducting bodies of the characteristic size $a$,
 $x_m \in D_m$, $D:=\bigcup_{m=1}^M D_m$, $D_m \subset \Omega, D'=\RRR \setminus D$. Assume that $D_m$ are distributed
 in a bounded domain $\Omega$ according to the formula
\be \label{eq2.26}
	\mathcal{N}(\Delta)=\frac{1}{a^3} \int_{\Delta} N(x)dx(1+o(1)), \quad a\to 0,
\ee
where $\Delta\subset \Omega$ is an arbitrary open subset of $\Omega$,  $N(x)\ge 0$ is a continuous in $\Omega$ function which
can be chosen by the experimenter as he/she wishes.
Let us assume that relation \eqref{eq1.1} holds. If $\Omega$ is a cube with the size $L$, then
$$\left(\frac{L}{d}\right)^3=O(M)=O(\frac{1}{a^3}),$$
so $d=O(La)$. Therefore condition $d\gg a$ can hold if $L$ is sufficiently large. If $L$ is fixed, then the condition $d \gg a$
can hold if
$N \ll 1$, because under this assumption about $N$ one has:
  $$d=O\left(\frac{a}{\left(\int_\Omega N(x)dx\right)^{1/3}}\right)\gg a.$$

The many-body scattering problem consists of solving equations \eqref{eq2.1} with $D=\bigcup_{m=1}^M D_m$, with boundary
conditions \eqref{eq2.2}, where $S=\bigcup_{m=1}^M S_m$, and with radiation condition \eqref{eq2.5}. The solution to this
 problem is unique.

We look for the solution of the form
\be \label{eq2.27}
	E=E_0+\sum_{m=1}^{M} \nabla\times \int_{S_m} g(x,t)J_m(t)dt.
\ee
This formula can be written as
\be \label{eq2.28}
	E=E_0+\sum_{m=1}^{M} [\nabla g(x,x_m),Q_m]+f, \quad Q_m:=\int_{S_m} J_m(t)dt,
\ee
where
\be \label{eq2.29}
	f:=\sum_{m=1}^{M}\nabla\times \int_{S_m} (g(x,t)-g(x,x_m)J_m(t)dt:=
\sum_{m=1}^{M} f_m
\ee
Let us show that for all $m$ one has
\be \label{eq2.30}
	|f_m|\ll |I_m|:=|[\nabla g(x,x_m), Q_m]|, \quad a\to 0.
\ee
If \eqref{eq2.30} holds, then the asymptotically exact solution of the many-body scattering problem is of the form
\be \label{eq2.31}
	E=E_0+\sum_{m=1}^{M} [\nabla g(x,x_m),Q_m], \qquad a\to 0.
\ee
{\em This is a basic result: it reduces the solution to the many-body scattering problem to finding quantities $Q_m$ rather
 than to finding the vector-functions $J_m(t)$. Such a reduction makes it possible to solve the many-body scattering problem for
 so many particles that
 it was not possible to do earlier.}

Our assumption is $a\ll d\ll \lambda$. Since $k=\frac {2\pi}{\lambda}$, it follows that
$a\ll d\ll k^{-1}$.

 To check inequality \eqref{eq2.30}, note that
\begin{align}
	&|\nabla g(x,x_m)| \le O\left((k+d^{-1}) \frac{1}{d}\right)=O(\frac{1}{d^2}), \qquad |x-x_m|=d, \label{eq2.32} \\
	&|\nabla g(x,t)-\nabla g(x,x_m)| \le O\left(a(k+d^{-1})\frac{1}{d^2}\right), \quad |t-x_m|\le a. \label{eq2.33}
\end{align}
Thus, $|I_m|=O\left(|Q_m|\frac {1}{d^2}\right)$, $|f_m| \le O\left(|Q_m|a(k+d^{-1})\frac{1}{d^2}\right)$,
and $Q_m\neq 0$.
Consequently,
\be \label{eq2.34}
 \left|\frac{f_m}{I_m}\right|\le O\left(ka+ad^{-1}\right)\ll 1.
\ee
Note that our basic physical assumption  $a\ll d\ll \lambda$ implies $ka\ll ad^{-1}$ because
$k=\frac{2\pi}{\lambda}$, so $k\ll \frac 1 d$ and $ka\ll \frac a d$.

Let us define the notion of the {\em effective field $E_e$} acting on the j-th particle:
\be \label{eq2.35}
	E_e:=E_0(x)+\sum_{m\neq j}^{M} \nabla\times\int_{S_m} g(x,t)J_m(t)dt.
\ee
As $a\to 0$, the effective field is asymptotically equal to the full field because the radiation from one particle is
proportional to $O(a^3)$, See Theorem 2.3 in Section 2.1.

If
\be \label{eq2.36}
	ka+\frac{a}{d}\ll 1,
\ee
then, with the error negligible as $a\to 0$, one has
\be\label{eq2.37}
E=E_0+\sum_{m=1}^{M} [\nabla g(x,x_m),Q_m],
\ee
where
\be\label{eq2.38}
Q_m = -a^3 c_{D_m}(I+\Gamma)^{-1}\nabla\times E_e(x_m).
\ee
If the quantities $A_m:=(I+\Gamma)^{-1}(\nabla\times E_e)(x_m), 1 \le m\le M,$
are found, then the solution of the
many-body scattering problem for perfectly conducting small bodies of an arbitrary shape can be found by
formulas \eqref{eq2.37}-\eqref{eq2.38}.

The shape of the small bodies enters only through the constants $c_{D_m}$, since $|D_m| =c_{D_m}a^3$.

In order to solve many-body scattering problem one needs to find the quantities $A_m$.
Let us reduce the problem of finding $A_m$ and $E_m$ to solving linear algebraic systems (LAS).

Put $x = x_j, E_0(x_j) := E_{0j}$ in \eqref{eq2.37}, assume for simplicity that $c_{D_m} = c_D$ for all $m$,
that is, the small bodies are identical, let $m \neq j$ and get
\be\label{eq2.52f}
E_j = E_{0j} - c_D\sum_{m \neq j}^M \left[ \nabla g(x_j, x_m),  A_m  \right]a^3, \quad 1 \leq j \leq M.
\ee
There are $2M$ vector unknowns $A_m$ and $E_m$ in this LAS and $M$ equations. One needs another set of $M$ linear
equations for finding these unknowns.

 To derive these equations, apply the operator $(I+\Gamma)^{-1}\nabla \times$ to  equation \eqref{eq2.37},
  denote $A_{0j}:=(I+\Gamma)^{-1} (\nabla\times E_0)(x_j)$,
 and set $x=x_j, j \neq m,$ in the resulting equation. This yields a linear algebraic system (LAS):
\be \label{eq2.39}
	A_j=A_{0j}-c_Da^3\sum_{m\neq j}^{M}\Big((I+\Gamma)^{-1}\nabla_x \times [\nabla g(x,x_m), A_m  ]\Big)|_{x=x_j},
 \quad 1\le j\le M.
\ee
Formula
$$\nabla_x \times [F(x), A]=(A,\nabla)F(x)- A (\nabla, F(x)),$$
valid if the vector $A$ is independent of $x$, can be useful.

If $M$ is very large, then the order of the LAS \eqref{eq2.52f} - \eqref{eq2.39} can be {\em drastically reduced} by
 the following method.

Let\,\,  $\bigcup_{p=1}^P \Delta_p$\,\, be a partition of $\Omega$ into a union of cubes $\Delta_p$ with the
 side $b=b(a)$, $\lim_{a\to 0}b(a)=0$.
Assume that
\be \label{eq2.40}
	b \gg d \gg a, \quad \lim_{a\to 0} \frac{d}{b}=0.
\ee
At the points $x_m \in \Delta_p$ the values of $A_m$ and of $\nabla g(x,x_m)$, where $x \not\in \Delta_p$,
 are asymptotically equal as $a\to 0$. Therefore, equation \eqref{eq2.39} can be rewritten as
\be \label{eq2.41}
	A_q=A_{0q}-c_D\sum_{p\neq q} \Big((I+\Gamma)^{-1}\nabla_x \times[\nabla g(x,x_p), A_p ]\Big)|_{x=x_q}a^3 \sum_{x_m \in \Delta_p}1,
\ee
and equations \eqref{eq2.52f} can be transformed similarly. Here $x_p \in \Delta_p$  is an arbitrary point,
 $D_m \subset \Delta_p$, $x_m\in \Delta_p$, $D_m$  are small bodies in $\Delta_p$. Since $\Delta_p$ is small the quantities
$A_m$, $E_m$ and $g(x_j,x_m)$ for $x_m$ in $\Delta_p$ and $x_j\in \Delta_q$, $p\neq q$, are equal to $A_p$, $E_p$ and
 $g(x_q,x_p)$ respectively, up to the quantities of higher
order of smallness as $a\to 0$.

 By \eqref{eq2.26} one has
\be \label{eq2.42}
	a^3 \sum_{x_m \in \Delta_p}1=a^3\mathcal{N}(\Delta_p)=N(x_p)|\Delta_p|, \quad a\to 0,
\ee
where $|\Delta_p|$ is the volume of $\Delta_p$. Thus,
\be\label{eq2.57f}
E_q = E_{0q} -c_D \sum_{p\neq q}[\nabla g(x_q, x_p),  A_p ]N(x_p)|\Delta_p|, \quad 1 \leq q \leq P, \quad a \to 0,
\ee
\be \label{eq2.58f}
A_q=A_{0q}-c_D\sum_{p\neq q}^P \Big((I+\Gamma)^{-1}\nabla_x \times [\nabla g(x,x_p), A_p] N(x_p)|\Delta_p|\Big)|_{x=x_q},
 \quad 1\leq q \leq P,\quad a\to 0.
\ee
Equations \eqref{eq2.57f} - \eqref{eq2.58f}  is a LAS for $2P$ unknowns $A_{q}, E_q$, $P \ll M$.
Computational work can be considerably reduced if one solves first system \eqref{eq2.58f} for
$P$ unknown vectors $A_p$ and then calculate $P$
unknowns $E_p$ by formula \eqref{eq2.57f}.

 Since $P \ll M$,
 the order of the LAS \eqref{eq2.57f} - \eqref{eq2.58f} is {\em  much smaller than the order of
  LAS \eqref{eq2.52f} - \eqref{eq2.39}.}

A similar argument allows one to replace equation \eqref{eq2.37} by the following equation:
\be \label{eq2.45}
	E_{eq}=E_{0q}-c_D\Big(\nabla\times\sum_{p \neq q}^{P} g(x,x_p)((I+\Gamma)^{-1}\nabla\times E_e)(x_p) N(x_p)|\Delta_p|\Big)|_{x=x_q},
\ee
where the formula $[\nabla g(x), A]=\nabla \times (g A)$ was used. This formula is valid for a scalar function $g$ of $x$ and
a vector $A$, independent of $x$.

Formula \eqref{eq2.45} is a Riemannian sum for the following {\em limiting integral equation}:
\be \label{eq2.46}
	E(x)=E_{0}(x)-c_D\nabla\times\int_\Omega g(x,y)(I+\Gamma)^{-1}\nabla\times E(y)  N(y)dy.
\ee
The method used for the derivation of equation \eqref{eq2.46} in contrast
to the usual assumptions of the homogenization theory does not use periodicity assumption and the operator of
our problem does not have a
discrete spectrum.

Let us state our result.
\begin{thm} \label{thm2.2}
If assumptions \eqref{eq2.36} hold, then the unique solution to the many-body scattering problem can be calculated by formula
\eqref{eq2.31}, where $Q_m$ are given in \eqref{eq2.38} and $(I+\Gamma)^{-1}(\nabla\times E_e)(x_m):=A_m$ and $E_e(x_m) := E_m$ are found from the LAS \eqref{eq2.52f} - \eqref{eq2.39}. The order of LAS \eqref{eq2.52f} - \eqref{eq2.39} can be drastically reduced
 if assumptions \eqref{eq2.40} hold, and one obtains LAS \eqref{eq2.57f} - \eqref{eq2.58f}
of the order $P<< M$. As $a\to 0$, the electric field in the medium tends uniformly to the limit $E(x)$ which satisfies
equation \eqref{eq2.46}.
\end{thm}

 Apply the operator $\nabla\times\nabla\times$
 to equation \eqref{eq2.46} and use the formulas
  $$\nabla\times\nabla\times=\nabla\nabla\cdot-\nabla^2,\quad \nabla\cdot E=0,\quad  -\nabla^2 g=k^2g+\delta(x-y).$$
Assume for simplicity that $\Gamma$ is a diagonal matrix, $\Gamma :=\gamma I$, and let $C_D:= \frac {c_D}{1+\gamma}$.
Then
\be\label{eq2.61f}
\nabla \times \nabla \times E = k^2 E - C_D\nabla \times \left( N(x) \nabla \times E \right) =
  k^2E - C_D N(x) \nabla \times \nabla \times E -  C_D [\nabla N, \nabla \times E].
\ee
Consequently,
\be\label{eq2.62f}
\nabla \times \nabla \times E = \frac{k^2 E}{1 + C_D N(x)} - \frac{ C_D[\nabla N, \nabla \times E]}{1 + C_D N(x)}.
\ee
It is clear from \eqref{eq2.62f} that the refraction coefficient in the medium where many small perfectly conducting particles are
 distributed is changed: the new refraction coefficient is proportional to $\Big(1 + C_D N(x)\Big)^{-1}$. The second term
 on the right-hand side of equation \eqref{eq2.62f} can be interpreted as coming from the new magnetic permeability.
 Indeed, if $\mu=\mu(x)$ in Maxwell equations, then taking $\nabla \times$ of the first equation and using the
 second equation one gets $\nabla \times \nabla \times E=k^2E+ [\frac{\nabla \mu(x)}{\mu(x)}, \nabla \times E]$.
 Compare this formula with equation  \eqref{eq2.62f} and conclude that $\mu(x)=\Big(1 + C_D N(x)\Big)^{-1}$.

Since $\nabla \cdot E=0$, one has $\nabla \times \nabla \times E =\nabla \nabla\cdot E-\nabla^2 E=-\nabla^2 E$,
and since $N(x)\ge 0$ is compactly supported, equation \eqref{eq2.62f} is a Schr\"odinger-type equation
with compactly supported potential and the terms with the first derivatives, the coefficients
in front of which are compactly supported. The solution of this equation satisfies the radiation condition
at infinity.


\section{Scattering by one impedance particle of an arbitrary shape}\label{sec3}
The problem consists of finding the solution to the system \eqref{eq2.1}, assuming that $E = E_0 + v_E$, $E_0$ is
given in \eqref{eq2.3}, the scattered field $v_E$ satisfies the radiation condition \eqref{eq2.5}, and $E$ satisfies the
impedance boundary condition
\begin{equation}\label{eq3.1}
[N, [E, N]] = \zeta [N, H] \quad \text{ on } S, \quad Re\, \zeta \geq 0,
\end{equation}
where $\zeta$ is a number, the boundary impedance. We will use condition \eqref{eq3.1} in the form
\begin{equation}\label{eq3.2}
[N, [v_E, N]] - \frac{\zeta}{i\omega\mu}[N, \nabla \times v_E] = -f,
\end{equation}
where
\begin{equation}\label{eq3.3}
f := [N, [E_0, N]] - \frac{\zeta}{i\omega\mu}[N, \nabla \times E_0].
\end{equation}
Let us look for the solution of the scattering problem with the impedance boundary condition in the form \eqref{eq2.8} where
$J(t)$ is a tangential field to $S$.

{\em It was not known if this solution can be represented in the form \eqref{eq2.8}. We prove in this Section that
one can find the solution in the form \eqref{eq2.8} and the scattered
field can be found  in the form \eqref{eq3.4}, see below.}

 The uniqueness of the solution to the EM wave scattering problem by an impedance body is known (see a proof in \cite{R635}).
  The existence of the solution in the form involving a sum of four boundary integrals was known (see \cite{CK}),
  but such a representation of the solution is not useful for our purposes. We want to give an explicit closed-form
  formula for the field scattered by a small impedance particle of an arbitrary shape.

The integral equation for $J$, which one gets by substituting
\begin{equation}\label{eq3.4}
v_E = \nabla \times \int_S g(x,t)J(t)dt
\end{equation}
into boundary condition \eqref{eq3.2}, is not of a Fredholm class: it is a singular integral equation.

Our approach to solving the scattering problem for a small impedance particle can be described as follows. We prove that
 its solution exists and can be represented in the form \eqref{eq3.4} by using the general theory of elliptic
 systems (see \cite{S}) and checking that {\em the complementing or covering condition, also known as
  Lopatinsky-Shapiro (LS) condition,} is satisfied (see \cite{R651}).

  Note that if the solution exists, it can be found in the form \eqref{eq3.4}. Indeed, one can calculate $[N,e]$ on $S$,
  and solve the problem for perfectly conducting particle with the boundary condition    $[N,e]$ on $S$. If $[N,e]$ on $S$
  is known, then $e$ is uniquely determined, so the corresponding scattering problem is uniquely solvable
  and its solution, as follows from Theorem 2.1, can be found in the form  \eqref{eq3.4}.

 Next, we prove that asymptotically, as $a \to 0$, the main term in the scattered field is given by the formula
\begin{equation}\label{eq3.5}
v_E = [\nabla g(x,x_1), Q], \qquad a \to 0,
\end{equation}
where $x_1 \in D$ is an arbitrary point inside the small particle $D$, and
\begin{equation}\label{eq3.6}
Q := \int_S J(t)dt.
\end{equation}
{\em This is an important point: not the function $J(t)$ but just the quantity $Q$ defines main term of the scattered
field if the body $D$ is small, $ka \ll 1$. From the physical point of view solving the scattering problem is reduced
 to finding vector $Q$ rather than the vector-function $J(t)$.
From the numerical point of view such a reduction makes it possible to solve scattering problems with so many small
 particles that was
impossible to solve earlier.}

Finally, we give, as $a \to 0$,  a formula for $Q$:
\begin{equation}\label{eq3.7}
Q = - \frac{\zeta |S|}{i\omega\mu}\tau_1 \nabla \times E_0,
\end{equation}
see formula  \eqref{eq3.304} below, where $\tau_1:=(I+\Gamma)^{-1} \tau$, and $\tau$ is defined in formula \eqref{eq3.8}.

In formula \eqref{eq3.7}  $|S|$ is the surface area of $S := \partial D$, $\zeta$ is the boundary impedance
(see condition \eqref{eq3.2}), and the tensor $\tau$ is defined as follows:
\begin{equation}\label{eq3.8}
\tau_{jp} := \delta_{jp} - b_{jp}, \quad b_{jp} := \frac{1}{|S|}\int_S N_j(t)N_p(t)dt.
\end{equation}
 Formulas \eqref{eq3.5}, \eqref{eq3.7} and \eqref{eq3.8} solve the EM wave scattering problem for a small
 impedance body of an arbitrary shape.
It follows from  formula \eqref{eq3.7} that $Q=O(a^{2-\kappa})$ because $|S|=O(a^2)$ and $\zeta=O(a^{-\kappa})$
 as $a\to 0$.

Let us prove these statements. We start with the uniqueness and  existence of the solution of the scattering problem with
the impedance boundary condition.

Uniqueness of this solution is known (see \cite{R635}, p. 81). Let us reduce solving Maxwell's system \eqref{eq2.1} to an
 equivalent elliptic system for $E$. If $E$ is found then $H$ is given by the formula
\begin{equation}\label{eq3.9}
H = \frac{\nabla \times E}{i\omega\mu}.
\end{equation}
Assume that  $\mu=$ const in $D'$.  Apply the operator $\nabla \times$ to the first equation \eqref{eq2.1} and use the second
 equation \eqref{eq2.1} to get:
\begin{equation}\label{eq3.10}
\nabla \times \nabla \times E = k^2E, \quad \nabla \cdot E = 0, \quad \text{ in } D',
\end{equation}
where $k^2 = \omega^2\epsilon\mu$.
Equations \eqref{eq3.10} imply
\begin{equation}\label{eq3.11}
\left( \nabla^2 + k^2 \right)E = 0, \qquad \nabla \cdot E = 0 \quad \text{ in } D'.
\end{equation}
Since $E_0$ solves equations \eqref{eq3.11} in $\mathbb{R}^3$, one concludes that
\begin{equation}\label{eq3.12}
\left( \nabla^2 + k^2 \right)v_E = 0, \qquad \nabla \cdot v_E = 0 \quad \text{ in } D'.
\end{equation}
Equations \eqref{eq3.12} could be replaced by one elliptic system:
\begin{equation}\label{eq3.13}
\left( -\nabla^2 - k^2 \right)v_E = 0 \qquad \text{ in } D',
\end{equation}
and the boundary condition
\begin{equation}\label{eq3.14}
\nabla \cdot v_E = 0 \qquad \text{ on } S.
\end{equation}
Indeed, the function $\psi(x):=\nabla \cdot v_E$ solves the problem
\begin{equation}\label{eq3.15}
\left( \nabla^2 + k^2 \right)\psi = 0 \quad \text{ in } D', \quad \psi|_S = 0,
\end{equation}
and $\psi$ satisfies the radiation condition \eqref{eq2.5}. This implies (see \cite{R190}, p. 28) that $\psi = 0$ in $D'$.

Therefore, our scattering problem is reduced to solving elliptic system \eqref{eq3.13} with boundary conditions
\eqref{eq3.14} and \eqref{eq3.2} and the radiation condition \eqref{eq2.5}.

Let $w(x):= (1+|x|^2)^{-\gamma}$, where $\gamma>\frac 1 2$, be a  weight function. This weight is chosen so that
the functions $v_E$, that are $O(\frac 1{|x|})$ as $|x|\to \infty$, belong to $L^2(D', w)$. By $H^2(D', w)$
the weighted Sobolev space is denoted.

\begin{theorem}\label{thm3.1}
The solution $v_E$ to the elliptic system \eqref{eq3.13} with boundary conditions \eqref{eq3.14} and \eqref{eq3.2}
and the radiation condition \eqref{eq2.5} exists in $H^2(D', w)$, is unique, $v_E=O(\frac 1{|x|})$ as $|x|\to \infty$,
 and $v_E$ can be found of the form \eqref{eq3.4}.
\end{theorem}
\begin{proof}
Clearly, system \eqref{eq3.13} is elliptic. Let us check that the LS (complementary) condition is satisfied.
The principal symbol of the operator \eqref{eq3.13} is $\xi^2 \delta_{pq}$, where $\xi$ is the parameter of the Fourier transform:
\begin{equation}\label{eq3.16}
u(x) = \int_{\mathbb{R}^3}\tilde{u}(\xi)e^{i\xi\cdot x}d\xi.
\end{equation}
If $D_j := -i\frac{\partial}{\partial x_j}$, then equation \eqref{eq3.13} can be rewritten as follows:
\begin{equation}\label{eq3.17}
\sum_{j = 1}^3 D_j^2v_E - k^2v_E = 0 \quad \text{ in } D'.
\end{equation}
The boundary conditions \eqref{eq3.2} and \eqref{eq3.14} can be written in the form
\begin{equation}\label{eq3.18}
B(D)v_E = F, \qquad F:=\left(\begin{array} {c}
f\\0
\end{array}\right),
\end{equation}
where $D:= (D_1, D_2, D_3)$, $f$ is a two-dimensional vector in the tangential to $S$ plane in the local coordinates and the
zero component in the vector $F$ in formula  \eqref{eq3.18} comes from the condition $\nabla \cdot v_E=0$ on $S$.
 The matrix $B(D)$ is defined by one vector boundary condition \eqref{eq3.2} and one scalar boundary condition
 \eqref{eq3.14}. In the local coordinates on $S$, in which the exterior unit normal $N$ to $S$ is directed along
  the $z-$axis, one has  $N = (0,0,1)$, and
 the principal symbol of the boundary matrix differential operator $B(D)$ is:
\begin{equation}\label{eq3.19}
B(\xi) := \frac{\zeta}{i\omega\mu} \left( \begin{array}{ccc}
-i\xi_3 & 0 & i\xi_1 \\
0 & -i\xi_3 & i\xi_2 \\
i\xi_1 & i\xi_2 & i\xi_3 \end{array} \right).
\end{equation}
The operator $D_j$ is mapped by the Fourier transform \eqref{eq3.16} onto $\xi_j$.

Let  $D_t := D_3$. The LS condition holds if the following problem
\begin{equation}\label{eq3.20}
\left( -\frac{d^2}{dt^2} + \rho^2 \right)u(\xi_1, \xi_2, t) = 0, \quad t > 0, \quad \rho^2 := \xi_1^2 + \xi_2^2,
\end{equation}
\begin{equation}\label{eq3.21}
B(\xi_1, \xi_2, D_t)u(\xi_1, \xi_2, t)|_{t=0} = 0,
\end{equation}
has only the zero solution, provided that one uses exponentially decreasing, as $t \to \infty$, solution of
equation \eqref{eq3.20}, that is, $u = e^{-t\rho}v$, $v = v(\xi_1, \xi_2) = (v_1, v_2, v_3)$, see \cite{S}.

Therefore, the LS condition holds if and only if the matrix
\begin{equation}\label{eq3.22}
\left( \begin{array}{ccc}
-i\rho & 0 & i\xi_1 \\
0 & -i\rho & i\xi_2 \\
i\xi_1 & i\xi_2 & i\rho \end{array} \right)
\end{equation}
is non-degenerate for $\rho > 0$. The determinant of this matrix equals to
\begin{equation}\label{eq3.22a}
-i\rho(\rho^2 + \xi_1^2 + \xi_2^2) \neq 0 \qquad \text{\,\, if\,\,} \rho > 0.
\end{equation}
 Thus, the LS condition holds. This implies the Fredholm property of the corresponding problem in the spaces
 $H^m(D', w)$ where $w = (1 + |x|^2)^{-\gamma}$, $\gamma > \frac{1}{2}$, that is, in the weighted Sobolev spaces with the norm
  $||v||_m^2 := \int_{D'}\sum_{l = 0}^m|D^l v|^2 w(x)dx$.
  The weight $w$ is chosen so that the functions decaying as $O(|x|^{-1})$ at infinity belong to $H^m(D',w)$.
 Since the LS condition holds,
 the elliptic estimate holds for the solution to problem \eqref{eq3.13}, \eqref{eq3.14}, \eqref{eq3.2}, \eqref{eq2.5}:
\begin{multline}\label{eq3.23}
||v_E||_{m + 2} \leq c\left( ||(\nabla^2 + k^2)v_E||_m + |B(D)v_E|_{m + \frac{1}{2}} + ||\eta v_E||_0 \right) \\
\leq c\left( |f|_{m + \frac{1}{2}} + ||\eta v_E||_0 \right),
\end{multline}
where $\eta$ is a smooth non-negative cut-off function vanishing near infinity,  $||v||_m$ is the norm in
$H^m(D', w)$ and $|v|_m$ is the norm in the
Sobolev space $H^m(S)$ on the boundary $S$,  see \cite{S}.

Due to the uniqueness of the solution to the scattering problem one can reduce estimate  \eqref{eq3.23} to the following
estimate:
\begin{equation}\label{eq3.23'}
  ||v_E||_{m+2}\le c|f|_{m+\frac {1} {2}},
\end{equation}
where $c>0$ here and below denotes various estimation constants.
To prove estimate  \eqref{eq3.23'} assume that it is false and derive a contradiction. If estimate \eqref{eq3.23'} is false,
then there is a sequence $v_{En}$, $||v_{En}||_{m+2}=1$,
such that
\begin{equation}\label{eq3.23a}
||v_{En}||_{m+2}\ge n |f_n|_{m+\frac {1} {2}}.
\end{equation}
Therefore, in any compact subdomain $D"$ of $D'$ one can select
a convergent in $H^{l}(D")$, $l<m$, subsequence which we denote again $v_{En}$. Assume  for concreteness that $m=0$.
Then, by the Sobolev embedding
theorem, $v_{En}$ converges strongly in $H^l(D")$ for $l<2$. Estimate  \eqref{eq3.23} implies that
$$||v_{Ej}-v_{Em}||_{H^2(D")}\le c(|f_j-f_m|_{1/2}+||\eta (v_{Ej}-v_{Em})||_{H^0(D")}\to 0 \text{\,\, as \,\,} j,m\to \infty.$$
Thus, $v_{En}$ converges in $H^2(D")$ to some element $v$, $||v||_2=1$. It follows from
estimate  \eqref{eq3.23a} and from the relation  $||v_{En}||_{2}=1$ that $|f_n|_{\frac {1} {2}}\to 0$ as $n\to \infty$.
Let us check that  $v$ satisfies the radiation condition. This is done as follows. Denote $v_{En}:=v_n$ and write the Green's formula:
\begin{equation}\label{eq3.23"}
v_n(x)=\int_S\Big(\frac{\partial g(x,s)}{\partial N}v_n(s)-g(x,s)\frac {\partial v_n}{\partial N} \Big)ds.
\end{equation}
Pass to the limit $n\to \infty$ in this formula. This is possible since, by the Sobolev embedding theorem,
 the embedding of $H^l(D")$ into $H^1(S)$
 is compact if $l>\frac 3 2$ provided that $D"\subset \mathbb{R}^3$, see \cite{KA}.
   Due to the local convergence in $H^l(D")$, $\frac {3}{2}<l<2$, one can pass to the limit $n\to \infty$
   in equation \eqref{eq3.23"} and get
  \begin{equation}\label{eq3.23b}
v(x)=\int_S\Big(\frac{\partial g(x,s)}{\partial N}v(s)-g(x,s)\frac {\partial v}{\partial N} \Big)ds.
\end{equation}
This implies that $v$ satisfies the radiation condition.

Therefore $v$  solves the homogeneous scattering problem and, by the uniqueness of the solution to this
problem, $v=0$. This contradicts the normalization  $||v||_2=1$, and the contradiction proves estimate   \eqref{eq3.23'}.

 {\em The index of our problem is zero}.

 This follows from the uniqueness of the solution to the homogeneous version of the scattering
 problem \eqref{eq3.13}, \eqref{eq3.14}, \eqref{eq3.2}, \eqref{eq2.5}, see also Lemma 2.1.

 Equation \eqref{eq3.2} can be written as
 \begin{equation}\label{eq3.2'}
 v_{E\tau}=\frac {\zeta}{i\omega \mu}V([N, \nabla \times v_{E\tau}])-V(f),
 \end{equation}
 where the operator $V$ was introduced in formula \eqref{eq2.192}, and $ v_{E\tau}$ is the tangential component of $v_E$.
 Let us assume that $f\in H^m(S)$. If $v_{E\tau}\in H^m(S)$,
 then $\nabla \times v_{E\tau}\in H^{m-1}(S)$. Therefore, it follows from  equation \eqref{eq3.2'} that
 $V([N, \nabla \times v_{E\tau}])\in  H^m(S)$. This means that $V$ acts from $ H^{m-1}(S)$ into $ H^{m}(S)$. Since the embedding
from  $ H^{m}(S)$ into  $ H^{m-1}(S)$ is compact,  $V$ is compact in  $ H^{m}(S)$.

  We have proved the existence of the unique solution to problem \eqref{eq3.13}, \eqref{eq3.14}, \eqref{eq3.2}, \eqref{eq2.5}.
  This problem is equivalent to the scattering problem \eqref{eq2.1}, \eqref{eq3.1}, \eqref{eq2.3'}, \eqref{eq2.5}.

Let us prove that if a solution to this scattering problem exists, then the scattered field $v_E$ can
be represented in the form \eqref{eq3.4}.

Let $E$ solve problem \eqref{eq2.1}, \eqref{eq3.1}, \eqref{eq2.3'}, and \eqref{eq2.5}. The tangential component
 $[N, [E, N]]$ on $S$ determines uniquely $E$ in $D'$. There is a one-to-one correspondence between $E$ and $v_E$,
  where $v_E = E - E_0$, and $v_E$ satisfies the boundary condition \eqref{eq3.2} with $f$ defined in \eqref{eq3.3}.
  The $v_E$ of the form \eqref{eq3.4} can be found from equation of the type \eqref{eq2.12}.
 Theorem \ref{thm2.1} guarantees that this equation is solvable for $J$ and the solution is unique.
  The corresponding $v_E$, defined by formula \eqref{eq3.4}, is the scattered field, and $E = E_0 + v_E$ is
  the unique solution to the scattering problem \eqref{eq2.1}, \eqref{eq3.1}, \eqref{eq2.3'}, \eqref{eq2.5}.

Theorem \ref{thm3.1} is proved.
\end{proof}

\begin{corollary}\label{cor3.1}
The smoothness of $v_E$ is $\frac{3}{2}$ derivatives more than the smoothness of the data $f$, as follows from the
 estimate \eqref{eq3.23}.
\end{corollary}

\begin{lemma}\label{lem3.1}
Formula \eqref{eq3.5} is asymptotically exact.
\end{lemma}
\begin{proof}
The proof is similar to the proof of formula \eqref{eq2.31}. Namely, one has
\begin{equation}\label{eq3.24}
E = E_0 + [\nabla g(x, x_1), Q] + \nabla \times \int_S\left( g(x,x_1) - g(x,t) \right)J(t)dt,
\end{equation}
where $x_1 \in D$ is an arbitrary point and
\begin{equation}\label{eq3.25}
Q := \int_S J(t)dt.
\end{equation}
Note that $g(x,x_1)=O(\frac 1 d)$, where $d := |x - x_1|$. When one differentiate $g$ one gets
\begin{equation}\label{eq3.26}
|\nabla g(x, x_1)| = O\Big( \frac 1 {d} (k + \frac{1}{d}) \Big), \qquad d := |x - x_1|.
\end{equation}
\begin{equation}\label{eq3.27}
\left| \nabla\left( g(x,x_1) - g(x,t) \right) \right| = O\left( (\frac{k}{d} + \frac{1}{d^2}) a (k+ \frac{1}{d}) \right),
\quad a = |x_1 - t|\ll d\ll k^{-1}.
\end{equation}
The quantity $Q$ does not vanish. Thus, the ratio of the third to the second term on the right-hand side of equation
 \eqref{eq3.24} is of the order
\begin{equation}\label{eq3.28}
O\left( ka + \frac{a}{d} \right) \ll 1.
\end{equation}
Lemma \ref{lem3.1} is proved.
\end{proof}

\begin{corollary}\label{cor3.2}
Formula \eqref{eq3.5} shows that solving the scattering problem by a small body ($ka \ll 1$) amounts to finding one quantity $Q$
 rather than the function $J(t)$ on $S$.
\end{corollary}

This is crucial for the solution of the many-body scattering problem that we present in Section \ref{sec4}.

{\bf Lemma 3.4.} {\em
Formula \eqref{eq3.7} holds as $a \to 0$.}

{\em Proof.} Proof of Lemma 3.4 is based on the following idea: we
take the vector product of $N_s$ with equation  \eqref{eq3.2}, then integrate the resulting equation over $S$ and keep the
 main term as $a \to 0$.
If
\begin{equation}\label{eq3.29}
\zeta = \frac{h}{a^\kappa}, \quad \kappa \in [0, 1), \quad Re\,\, h \geq 0,
\end{equation}
then one obtains
\begin{equation}\label{eq3.30}
Q = O\left( a^{2 - \kappa} \right), \quad a \to 0.
\end{equation}

Theorem \ref{thm3.1}  gives a mathematical justification of the
smoothness of $v_E$ and, therefore, of $J(t)$ provided that the data are smooth, see Corollary \ref{cor3.1}.
 This result is important for mathematical justification of the boundedness of the
second derivatives of the function $J(t)$, which is assumed but not  justified on p. 91 in \cite{R635}.
The estimates, necessary for a justification of formula  \eqref{eq3.7} are given on pp.88-93 in \cite{R635}.
The term $\int_Sds \int_S dt \nabla_s g(s,t) N(s)\cdot J(t)$ was neglected in \cite{R635}.
This term depends on a vector whose components are
$\int_S \Gamma_{pq}(t)J_q(t)dt$. Here and below over the repeated indices summation is understood
and  $\Gamma_{pq}(t):=\int_S \frac {\partial{g(s,t)}}{\partial{s_p}}N_q(s)ds$,
where the integral is understood as a singular integral.

 If one takes into account the term
\begin{equation}\label{eq3.301}
\int_Sds \int_S dt \nabla_s g(s,t) N(s)\cdot J(t)=e_p\int_S \Gamma_{pq}(t)J_q(t)dt,
\end{equation}
where $\{e_p\}_{p=1}^3$ is an orthonormal basis of $\mathbb{R}^3$, then in place of
equation  \eqref{eq3.7} one obtains the following equations:
\begin{equation}\label{eq3.302}
 \int_S J_p(t)dt +\int_S  \Gamma_{pq}(t)J_q(t)dt=- \frac{\zeta |S|}{i\omega \mu}(\tau \nabla \times E_0, e_p) , \qquad 1\le p\le 3.
 \end{equation}
There exists a constant matrix  $\Gamma:=(\Gamma_{pq})$ such that
\begin{equation}\label{eq3.303}
e_p\int_S  \Gamma_{pq}(t)J_q(t)dt=\Gamma Q,
 \end{equation}
provided that $Q\neq 0$, which is our case. Equation \eqref{eq3.302} in this case
takes the form $(I+\Gamma)Q=- \frac{\zeta |S|}{i\omega \mu}\tau  \nabla \times E_0$, and
the matrix $I+\Gamma$ is non-singular since $Q\neq 0$.  Therefore,
\begin{equation}\label{eq3.304}
Q=- \frac{\zeta |S|}{i\omega \mu} (I+\Gamma)^{-1}\tau  \nabla \times E_0.
 \end{equation}
Lemma 3.4 is proved.  \hfill $\Box$

 The many-body scattering problem
is discussed in the next Section on the basis of formula   \eqref{eq3.7}.
This is done for simplicity of notations, since formula \eqref{eq3.304}
can be identified with formula \eqref{eq3.7} if one replaces $\tau$ by
 $\tau_1:=(I+\Gamma)^{-1}\tau$.


\section{Many-body scattering problem}\label{sec4}
This problem consists of finding $E$ and $H = \displaystyle\frac{\nabla \times E}{i\omega\mu}$, which satisfy equations
 \eqref{eq2.1} with $D = \cup_{m=1}^M D_m \subset \Omega$, $E$ is of the form \eqref{eq2.3'} and satisfies the impedance boundary
  conditions on $S_m = \partial D_m$:
\begin{equation}\label{eq4.1}
[N, [E, N]] = \frac{\zeta_m}{i\omega\mu}[N, \nabla \times E] \quad \text{ on } S_m; \quad Re\,\,\zeta_m \geq 0,
\end{equation}
and the radiation condition \eqref{eq2.5} for the scattered field $v_E$. We look for $v_E$ of the form
\begin{equation}\label{eq4.2}
v_E = \sum_{m = 1}^M \nabla \times \int_{S_m}g(x,t)J_m(t)dt, \quad E = E_0 + v_E,
\end{equation}
where $J_m$ is a tangential to $S_m$ field.

{\em The basic physical (and mathematical) assumptions are \eqref{eq1.1} and \eqref{eq1.2}.}

{\em The basic results of this section can be described as follows:}
\begin{itemize}
\item[1.] The above EM wave scattering problem has a solution, this solution is unique and can be found in the form \eqref{eq4.2}.
\item[2.] As $a \to 0$, the main term of the solution to the EM wave scattering problem is
\begin{equation}\label{eq4.3}
E = E_0 + \sum_{m = 1}^M\left[\nabla g(x,x_m), Q_m \right], \quad a \to 0; \quad Q_m := \int_{S_m}J_m(t)dt,
\end{equation}
where $x_m \in D_m$ are arbitrary points.
\item[3.] An explicit, asymptotically exact as $a \to 0$, formula for $Q_m$ is derived:
\begin{equation}\label{eq4.4}
Q_m = - \frac{\zeta_m|S_m|}{i\omega\mu}\tau_m(\nabla \times E_e)(x_m), \quad 1 \leq m \leq M,
\end{equation}
where $|S_m|$ is the surface area of $S_m$, $\zeta_m = \frac{h(x_m)}{a^\kappa}$, $Re\,\, h \geq 0$, where $h \in C(\Omega)$
 is a function the experimenter may choose as desired as well as the parameter $\kappa$, $\kappa \in [0,1)$, $\tau_m$ is the
 tensor defined by formula \eqref{eq3.8} with $S = S_m$, and $E_e(x)$ is {\em the effective field acting on the particle
  $D_m$}:\\
\begin{equation}\label{eq4.5}
E_e(x) := E_0(x) + \sum_{p \neq m}^M \nabla\times\int_{S_p}g(x,t)J_p(t)dt.
\end{equation}
Equation \eqref{eq4.5} is valid  not only in a neighborhood of $x_m$.
The field scattered by $m-$th particle is proportional to $a^{2-\kappa}$ and is negligible compared with $E_e(x)$
at any point $x$.

\item[4.] Derivation of a linear algebraic system (LAS) for calculating $Q_m$.
\item[5.] Proof of the existence of the limit $E(x)$ of the effective field $E_e(x)$ as $a \to 0$ and the derivation of the equation
for the limiting field $E(x)$.
\item[6.] Physical interpretation of the equation for the limiting field $E(x)$. Explicit formulas for the new refraction coefficient
 and magnetic permeability.
\end{itemize}

The uniqueness and existence of the solution are proved  similarly to the proof given in the case of the scattering problem for one
 body. Formulas \eqref{eq4.3} and \eqref{eq4.4} are established as in our theory of EM wave scattering by one body. An important point is the following one:

{\em Each of the $M$ small bodies can be considered under our basic assumption \eqref{eq1.1} as a single scatterer on which the
 incident field $E_e(x)$ is scattered. Therefore formula \eqref{eq3.7} remains valid after replacing $E_0$ by $E_e$, and this
 yields formula \eqref{eq4.4}.}

 Formula \eqref{eq4.3} is derived along the same lines as formula \eqref{eq2.37}. If
 $$(\nabla\times E_e)(x_m) := A_m, \qquad (\nabla\times E_0)(x_m) := A_{0m},\qquad \text{\,\, and\,\,} |S_m| = c_ma^2,$$
  then equations \eqref{eq4.3}-\eqref{eq4.5} imply
\begin{equation}\label{eq4.6}
A_j = A_{0j} - \Big(\nabla\times\sum_{j \neq m}^M \left[ \nabla g(x, x_m), \frac{h(x_m) c_m a^{2-\kappa}}
{i\omega\mu}\tau_m A_m \right]\Big)|_{x=x_j}, \quad 1 \leq j \leq M.
\end{equation}
This is a LAS for finding $A_m$. If $A_m$ are found, then
\begin{equation}\label{eq4.7}
Q_m = -\frac{h(x_m) c_m a^{2-\kappa}}{i\omega\mu}\tau_m A_m.
\end{equation}
For simplicity one may assume in what follows that $c_m = c_0$ and $\tau_m = \tau$ do not depend on $m$. One can write
equation \eqref{eq4.3} as
\begin{equation}\label{eq4.8}
E_e(x_j) = E_0(x_j) - \frac{c_0a^{2-\kappa}}{i\omega\mu}\Big(\sum_{j \neq m}^M\left[\nabla g(x,x_m), \tau(\nabla\times E_e)(x_m)
\right]h(x_m)\Big)|_{x=x_j}, \quad 1 \leq j \leq M.
\end{equation}
{\em The order of the LAS \eqref{eq4.6} and \eqref{eq4.8} can be drastically reduced.}

 Namely, consider a partition of $\Omega$ into a union of small cubes $\Delta_p$, $\displaystyle\cup_{p = 1}^P\Delta_p = \Omega$.
  Assume that the side $b = b(a)$ of $\Delta_p$ is much larger than $d$, $b>>d$,  so that there are many small bodies $D_m$ in
  every cube $\Delta_p$, and
\begin{equation}\label{eq4.9}
\lim_{a \to 0}b(a) = 0.
\end{equation}
Recall that $x_m \in D_m$ is a point inside $D_m$. Let $x_p \in \Delta_p$ be an arbitrary point. For all $x_m \in \Delta_p$ the
values $h(x_m) = h(x_p)$ up to the error that tends to zero as $a \to 0$, because $h$ is a continuous function and
$b(a) \to 0$ as $a \to 0$. The same is true for $\nabla g(x_j, x_m)$ and for $\tau(\nabla\times E_e)(x_m)$. Consequently,
\eqref{eq4.8} implies
\begin{multline}\label{eq4.10}
E_e(x_q) = E_0(x_q) - \frac{c_0}{i\omega\mu}\sum_{q \neq p}^P\left[\nabla g(x_q,x_p), \tau(\nabla\times E_e)(x_p)
\right]h(x_p) a^{2-\kappa}\sum_{x_m \in \Delta_p}1 \\ = E_0(x_q) - \frac{c_0}{i\omega\mu}\sum_{q \neq p}^P
\left[\nabla g(x_q,x_p), \tau(\nabla\times E_e)(x_p) \right]h(x_p)N(x_p)|\Delta_p|.
\end{multline}
Here we have used the assumption \eqref{eq1.2} in the form
\begin{equation}\label{eq4.11}
a^{2-\kappa}\sum_{x_m \in \Delta_p}1 = \int_{\Delta_p}N(x)dx\Big(1 + o(1)\Big) \approx N(x_p)|\Delta_p|,
\end{equation}
where $|\Delta_p|$ is the volume of $\Delta_p$.

Equation \eqref{eq4.10} is the Riemannian sum corresponding to the integral equation:
\begin{equation}\label{eq4.12}
E(x) = E_0(x) - \frac{c_0}{i\omega\mu}\nabla\times\int_{\Omega} g(x,y) h(y)N(y)\tau\nabla\times E(y) dy.
\end{equation}
Thus, the effective field $E_e$ has a limit $E$, as $a \to 0$, and this limit satisfies equation \eqref{eq4.12}.
We have proved the following Theorem.

{\bf Theorem 4.1.} {\em The effective field $E_e(x)$ in $\Omega$ tends to the limit $E(x)$ in $C(\Omega)$
and the limiting field $E(x)$ solves equation  \eqref{eq4.12}}

Let us {\em interpret physically equation \eqref{eq4.12}}. Let us apply the operator $\nabla\times\nabla\times$ to equation
\eqref{eq4.12}. This yields, after using the formulas $\nabla\times\nabla\times = \nabla\nabla\cdot - \nabla^2$ and
 $\nabla\cdot\nabla\times = 0$, the following equation:
\begin{equation}\label{eq4.13}
\nabla\times\nabla\times E = \nabla\times\nabla\times E_0 - \frac{c_0}{i\omega\mu}\nabla\times\int_\Omega
\left(-\nabla^2 g(x,y)\right)
h(y)N(y) \tau\nabla\times E(y) dy.
\end{equation}
Since $\nabla\times\nabla\times E_0 = k^2E_0$ and $-\nabla^2g(x,y) = k^2g(x,y) + \delta(x - y)$, equation \eqref{eq4.13}
can be written as follows:
\begin{equation}\label{eq4.14}
\nabla\times\nabla\times E = k^2E - \frac{c_0}{i\omega\mu}\nabla\times \Big(h(x)N(x) \tau\nabla\times E(x) \Big).
\end{equation}
Assume that $\tau$ is a {\em diagonal} tensor. For example, if $D_m$ are balls, then $\tau_{pq} = \frac{2}{3}\delta_{pq}$,
so $\tau = \frac{2}{3}I$, where $I$ is the unit tensor. In this case
\begin{equation}\label{eq4.15}
\nabla\times\Big(  hN \tau \nabla\times E  \Big) = \frac{2}{3}hN\nabla\times\nabla\times E + \frac{2}{3}\left[ \nabla(hN),
\nabla\times E \right].
\end{equation}
Therefore, in this case equation \eqref{eq4.14} can be rewritten as follows:
\begin{equation}\label{eq4.16}
\nabla\times\nabla\times E = \frac{k^2}{1 + \frac{2c_0}{3 i\omega\mu}h(x)N(x)} - \frac{2c_0}{3i\omega\mu}\,\,
 \frac{\left[ \nabla(hN), \nabla\times E \right]}{1 + \frac{2c_0}{3 i\omega\mu}h(x)N(x)}.
\end{equation}
The physical meaning of this equation becomes clear if one applies the operator $\nabla\times$ to the first equation \eqref{eq2.1}
assuming that $\mu = \mu(x)$, that is, assuming that $\mu$ is a function of $x$.

Then one gets
\begin{equation}\label{eq4.17}
\nabla\times\nabla\times E = i\omega\mu(x)\nabla\times H + i\omega[\nabla\mu(x), H].
\end{equation}
Using the second equation \eqref{eq2.1} one reduces \eqref{eq4.17} to the following equation
\begin{equation}\label{eq4.18}
\nabla\times\nabla\times E = k^2n^2(x)E + \left[ \frac{\nabla \mu}{\mu}, \nabla\times E \right], \qquad k^2 := \omega^2\epsilon\mu(x).
\end{equation}
Comparing \eqref{eq4.18} with \eqref{eq4.16} one concludes that the following Theorem is proved.

{\bf Theorem 4.2.}  {\em The refraction coefficient in the new limiting medium is given by the formula:
\begin{equation}\label{eq4.19}
n(x) = \frac{1}{\sqrt{1 + \frac{2c_0}{3 i\omega\mu}h(x)N(x)}},
\end{equation}
 and  the magnetic permeability in this medium is given by the formula:
\begin{equation}\label{eq4.20}
\mu(x) = \frac{\mu}{1 + \frac{2c_0}{3 i\omega\mu}h(x)N(x)},
\end{equation}
where $\mu = $const is the magnetic permeability in the original medium.}

 Note that according to formulas \eqref{eq4.16} and \eqref{eq4.18} one has:
\begin{equation}\label{eq4.21}
\frac{\nabla \mu(x)}{\mu(x)} = - \frac{2c_0}{3i\omega\mu}\frac{\nabla\left( h(x)N(x) \right)}{1 + \frac{2c_0}{3 i\omega\mu}h(x)N(x)}.
\end{equation}

\section{Creating materials with a desired refraction coefficient and a desired magnetic permeability}\label{sec5}
Formulas \eqref{eq4.19} and \eqref{eq4.20} allow one to give recipes for creating materials with a desired
refraction coefficient or a desired magnetic permeability.

Suppose that one wants to create a material with a desired refraction coefficient $n(x)$ by embedding in a given material
many small impedance particles. One has to choose a bounded domain $\Omega$,  where the small particles should be distributed,
and give a distribution law \eqref{eq1.2}
 for these particles in $\Omega$.  The function $N(x) \geq 0$ in   \eqref{eq1.2} can be chosen by the experimenter.
   Next, one has to give boundary impedances, defined by formula \eqref{eq1.1'}, where $h(x)$, $Re\,\, h \geq 0$, is a continuous
    in $\Omega$ function, which can also be chosen by the experimenter as he/she wishes, as well as the parameter $\kappa \in [0,1)$.

Let us prove the following Theorem:

{\bf Theorem 5.1.} {\em   Any refraction coefficient $n(x)$ can be obtained by choosing a suitable $h(x)$.}

 {\em Proof.} Suppose that
 $$h = h_1(x) + ih_2(x),\quad h_1(x) := Re\, h \geq 0,\quad  N(x) = N =  const, \quad \frac{2c_0 N}{3\omega\mu} := c_1 > 0.$$
  Then formula \eqref{eq4.19} yields
\begin{equation}\label{eq4.22}
n(x) = \frac{1}{\sqrt{1 - ic_1h_1(x) + c_1h_2(x)}}.
\end{equation}
Let us define
$$\sqrt{z} = |z|^{1/2}e^{i\frac{\varphi}{2}}, \quad \varphi = \arg z, \quad 0 \leq \varphi \leq 2\pi.$$
 Since $h_1 \geq 0$ and $h_2$ are arbitrary real-valued functions, let us denote
 $$u(x):=1 + c_1h_2(x), \qquad v(x):=c_1h_1(x), $$
 and write:
\begin{equation}\label{eq4.23}
\frac{1}{\sqrt{1 + c_1h_2(x) - ic_1h_1(x)}} = \frac{1}{\sqrt{u^2 + v^2(x)}}e^{-\frac{i}{2}\arg \big(1 + c_1h_2(x) - ic_1h_1(x)\big)}.
\end{equation}
If $|u|$ and $|v|$ are arbitrary, so is $\frac{1}{\sqrt{u^2 + v^2}}$. The argument $\varphi$ of $1 + c_1h_2(x) - ic_1h_1
\in (\pi, 2\pi)$ if $h_1 \geq 0$, so $-\frac{\varphi}{2} \in \left( -\frac{\pi}{2}, -\pi \right)$.
Choosing  $u$ and $v$ suitably one can get a desirable  amplitude  $\frac{1}{\sqrt{u^2 + v^2}}$ of the refraction coefficient
  and a desirable phase of it.

Theorem 5.1 is proved. \hfill$\Box$

{\bf Example.}  {\em If $-\frac{\varphi}{2} \approx -\pi$ then $Re\,\, n(x) < 0$
and $Im\,\, n(x) < 0$ can be made as small as one wishes, so it will be negligible.
Thus, the obtained material has {\em negative refraction}: the phase velocity is directed opposite to the group velocity
in this material.
 Recall that the phase velocity is $v_p = \frac{\omega}{|k|}\frac{k}{|k|}$, while the group velocity is $v_g = \nabla_k\omega(k)$.}

 Similar reasoning leads to a conclusion that {\em a desired magnetic permeability can also be created.}

 To do this one uses formula \eqref{eq4.19}.  Indeed,
\begin{equation}\label{eq4.24}
\mu(x) = \frac{\mu}{u(x) - iv(x)} = \frac{\mu}{\sqrt{u^2(x) + v^2(x)}}e^{-i\varphi}, \quad \varphi \in (\pi, 2\pi).
\end{equation}
The quantity $\frac{1}{\sqrt{u^2(x) + v^2(x)}}$ can be made arbitrary if $h_1(x) \geq 0$ and $h_2(x)$ can be chosen arbitrarily.
 The argument $\varphi \in (\pi, 2\pi)$ can be chosen arbitrarily.
\newpage
{\bf Remark.} {\em Principal differences of our results and the results of other authors on wave scattering by small
bodies are:

1. For wave scattering by one body: we derive a closed-form explicit formula for the scattering amplitude for small
bodies of an arbitrary shape for four types of the boundary conditions (the Dirichlet, the Neumann, the impedance, 
and the interface (transmission)), see \cite{R632},\cite{R635}.

2. For many-body wave scattering problems for small bodies of arbitrary shapes our condition $ka+ad^{-1}<<1$ allows
one to have $kd<<1$, that is, it allows to have many small particles on the wavelength. This means that the effective field
in the medium in which many small particles are distributed and the above conditions hold the effective field, acting on each small
particle, may differ very much from the incident field. That is,  the multiple scattering effects are essential and cannot
be neglected.

3. For solving problems of many-body wave scattering by small bodies an efficient numerical method is developed.

4. For many-body wave scattering problems the limiting equation for the effective field is derived
in the limit when the size of small impedance particles tends to zero while the number of these particles tends to infinity.  

5. A recipe is given for creating materials with a desired refraction coefficient by embedding many
small particles with prescribed boundary impedances into a given material.
}
\newpage

\bibliographystyle{apalike}

\end{document}